%% file: testing_domain.tex
\newcommand{\longversion}[1]{}
\newcommand{\shortversion}[1]{#1}

\longversion{\documentclass[a4paper,11pt]{article}}

\shortversion{
\documentclass[sigconf,usenames,svgnames,dvipsnames]{aamas}  

\usepackage{booktabs}

\setcopyright{ifaamas}  
\acmDOI{doi}  
\acmISBN{}  
\acmConference[AAMAS'19]{Proc.\@ of the 18th International Conference on Autonomous Agents and Multiagent Systems (AAMAS 2019), N.~Agmon, M.~E.~Taylor, E.~Elkind, M.~Veloso (eds.)}{May 2019}{Montreal, Canada}  
\acmYear{2019}  
\copyrightyear{2019}  
\acmPrice{}  
}

\author{Palash Dey}
\affiliation{%
  \institution{Indian Institute of Technology}
  \city{Kharagpur}
}
\email{palash.dey@cse.iitkgp.ac.in}
\author{Swaprava Nath}
\affiliation{%
  \institution{Indian Institute of Technology}
  \city{Kanpur}
}
\email{swaprava@cse.iitk.ac.in}
\author{Garima Shakya}
\affiliation{%
  \institution{Indian Institute of Technology}
  \city{Kanpur}
}
\email{garima@cse.iitk.ac.in}

\input{preamble}

\title{Testing Preferential Domains Using Sampling}

\title{{\bf Testing Preferential Domains using Sampling}}

\begin{document}

\longversion{
\maketitle
\input{abstract}
}

\shortversion{

\input{abstract}

\maketitle
}


\input{introduction.tex}
\input{prelim.tex}
\input{results.tex}
\input{empirical.tex}
\input{conclusion.tex}


\shortversion{\bibliographystyle{ACM-Reference-Format}}  
\longversion{\bibliographystyle{alpha}}
\bibliography{references}

\end{document}

%% file: preamble.tex
\usepackage{amssymb,amsmath,amsthm}
\usepackage{makecell}
\usepackage{multirow}
\usepackage{mathtools}
\usepackage{color}
\usepackage{xcolor}
\usepackage{url}
\usepackage{verbatim}
\usepackage{pdfpages}

\usepackage{enumerate}
\usepackage[shortlabels]{enumitem}

\longversion{
\usepackage{fullpage}
\usepackage{charter}
\usepackage{eulervm}
}

\usepackage{algpseudocode,algorithm,algorithmicx}

\algrenewcommand\algorithmicrequire{\textbf{Input:}}
\algrenewcommand\algorithmicensure{\textbf{Output:}}
\algnewcommand{\Initialize}[1]{%
  \State \textbf{Initialize:}
  \Statex \hspace*{\algorithmicindent}\parbox[t]{.8\linewidth}{\raggedright #1}
}
\algdef{SE}[FUNCTION]{Function}{EndFunction}%
   [2]{\algorithmicfunction\ \textproc{#1}\ifthenelse{\equal{#2}{}}{}{(#2)}}%
   {\algorithmicend\ \algorithmicfunction}%

\renewcommand{\Return}{\State \textbf{return} }

\newcommand{\randomtest}[4][\DD]{(#2,\xspace #3,\xspace #4,\xspace #1) -- {\sc Random Outliers vs Random Profile Test}}
\newcommand{\worstcasetest}[4][\DD]{(#2,\xspace #3,\xspace #4,\xspace #1) -- {\sc Arbitrary Outliers vs Random Profile Test}}
\newcommand{\randomtester}[4][\DD]{(#2,\xspace #3,\xspace #4,\xspace #1) -- {\sc Random Outliers vs Random Profile Tester}}
\newcommand{\worstcasetester}[4][\DD]{(#2,\xspace #3,\xspace #4,\xspace #1) -- {\sc Arbitrary Outliers vs Random Profile Tester}}
\newcommand{\worstworstcasetest}[6][\DD]{(#2,\xspace #3,\xspace #4,\xspace #5,\xspace #6,\xspace #1) -- {\sc Arbitrary Outliers vs Arbitrary Profile Test}}
\newcommand{\worstworstcasetester}[6][\DD]{(#2,\xspace #3,\xspace #4,\xspace #5,\xspace #6,\xspace #1) -- {\sc Arbitrary Outliers vs Arbitrary Profile Tester}}
\usepackage{url}

\makeatletter
\newcommand{\mathleft}{\@fleqntrue\@mathmargin\parindent}
\newcommand{\mathcenter}{\@fleqnfalse}
\makeatother

\usepackage{xspace}

\newcommand{\el}{\ensuremath{\ell}\xspace}
\newcommand{\suc}{\ensuremath{\succ}\xspace}
\let\mydelta\delta
\renewcommand{\delta}{\ensuremath{\mydelta}\xspace}

\renewcommand{\leq}{\leqslant}
\renewcommand{\geq}{\geqslant}
\renewcommand{\ge}{\geqslant}
\renewcommand{\le}{\leqslant}

\newcommand{\E}[1]{\mathbb{E}\left[#1\right]}

\newcommand{\pr}{\ensuremath{\prime}\xspace}
\newcommand{\prr}{{\ensuremath{\prime\prime}}\xspace}

\newcommand{\NPH}{\ensuremath{\mathsf{NP}}-hard\xspace}
\newcommand{\NPC}{\ensuremath{\mathsf{NP}}-complete\xspace}

\newcommand{\NB}{\ensuremath{\mathbb N}\xspace}

\newcommand{\EB}{\ensuremath{\mathbb E}\xspace}

\renewcommand{\AA}{\ensuremath{\mathcal A}\xspace}
\newcommand{\BB}{\ensuremath{\mathcal B}\xspace}

\newcommand{\DD}{\ensuremath{\mathcal D}\xspace}

\newcommand{\LL}{\ensuremath{\mathcal L}\xspace}

\newcommand{\OO}{\ensuremath{\mathcal O}\xspace}
\newcommand{\PP}{\ensuremath{\mathcal P}\xspace}
\newcommand{\QQ}{\ensuremath{\mathcal Q}\xspace}
\newcommand{\RR}{\ensuremath{\mathcal R}\xspace}
\renewcommand{\SS}{\ensuremath{\mathcal S}\xspace}

\newcommand{\WW}{\ensuremath{\mathcal W}\xspace}
\newcommand{\XX}{\ensuremath{\mathcal X}\xspace}

\newcommand{\ccc}{\ensuremath{\text{con}}\xspace}

\newcommand{\rrr}{\ensuremath{\text{res}}\xspace}

\usepackage{nicefrac}

\newcommand{\nfrac}{\nicefrac}

\newcommand{\eps}{\ensuremath{\varepsilon}\xspace}
\renewcommand{\epsilon}{\eps}

\usepackage{cleveref}

\longversion{\newtheorem{proposition}{\bf Proposition}}

\longversion{\newtheorem{theorem}{\bf Theorem}}
\longversion{\newtheorem{lemma}{\bf Lemma}}

\longversion{\newtheorem{corollary}{\bf Corollary}}
\longversion{}
\newtheorem{probdef}{\bf Problem}

\longversion{\newtheorem{example}{\bf Example}}

\crefname{example}{Example}{Examples}
\crefname{theorem}{Theorem}{Theorems}
\crefname{observation}{Observation}{Observations}
\crefname{lemma}{Lemma}{Lemma}
\crefname{corollary}{Corollary}{Corollaries}
\crefname{proposition}{Proposition}{Propositions}
\crefname{definition}{Definition}{Definitions}
\crefname{probdef}{Problem}{Problems}
\crefname{claim}{Claim}{Claims}
\crefname{reductionrule}{Reduction rule}{Reduction rules}
\crefname{ineq}{inequality}{Inequalities}


%% file: abstract.tex
\begin{abstract}
A preferential domain is a collection of sets of preferences which are linear orders over a set of alternatives. These domains have been studied extensively in social choice theory due to both its practical importance and theoretical elegance. Examples of some extensively studied preferential domains include single peaked, single crossing, Euclidean, etc. In this paper, we study the sample complexity of testing whether a given preference profile is close to some specific domain.
\longversion{ We consider various notion of closeness, for example, deletion of alternatives, deletion of preferences, and simultaneous deletion of alternatives as well as preferences.}\shortversion{ We consider two notions of closeness: (a) closeness via preferences, and (b) closeness via alternatives.} We further explore the effect of assuming that the {\em outlier} preferences/alternatives to be random (instead of arbitrary) on the sample complexity of the testing problem. In most cases, we show that the above testing problem can be solved with high probability for all commonly used domains by observing only a small number of samples (independent of the number of preferences, $n$, and often the number of alternatives, $m$). In the remaining few cases, we prove either impossibility results or $\Omega(n)$ lower bound on the sample complexity. We complement our theoretical findings with extensive simulations to figure out the actual constant factors of our asymptotic sample complexity bounds.
\end{abstract}

%% file: introduction.tex
\section{Introduction}

\begin{table*}[htb]
 \begin{center}
 \renewcommand{\arraystretch}{1.3}
  \begin{tabular}{|c|c|c|}\hline
  \multicolumn{2}{|c|}{Input profile} & \multirow{2}{*}{Sample complexity}\\\cline{1-2}
   Possibility $1$ & Possibility $2$ & \\\hline\hline
   
   
   \makecell{$\eps_v n$ random preferences away} & \multirow{\longversion{8}\shortversion{4}}{*}{random}  & \makecell{$\OO(\tfrac{1}{(1-\eps_v)^2}\log\tfrac{1}{\delta})$ [\Cref{thm:spc_n_test}]} \\\cline{1-1}\cline{3-3}
   
   \makecell{$\eps_v n$ arbitrary preferences away} &  & \makecell{$\OO(\tfrac{1}{(1-3\eps_v)^2}\ln\tfrac{1}{\delta})$ for $\eps_v<\nfrac{1}{3}^\star$ [\Cref{thm:spc_n_test_worst}]} \\\cline{1-1}\cline{3-3}
   
   $\eps_a m$ alternatives away &  & \makecell{$\OO(\log \tfrac{\log_{\nfrac{1}{\eps_a}}\nfrac{1}{\delta}}{\delta} \log_{\nfrac{1}{\eps_a}}\tfrac{1}{\delta}\log\log_{\nfrac{1}{\eps_a}}\nfrac{1}{\delta})$ [\Cref{thm:spc_m_test}]} \\\longversion{\cline{1-1}\cline{3-3}}\shortversion{\hline}
   
   \longversion{
   \makecell{$\eps_v n$ random preferences and $\eps_a m$ alternatives away} &  & \makecell{$\OO(\tfrac{1}{(1-\eps_v)^2}\log \tfrac{\log_{\nfrac{1}{\eps_a}}\nfrac{1}{\delta}}{\delta} \log_{\nfrac{1}{\eps_a}}\tfrac{1}{\delta}\log\log_{\nfrac{1}{\eps_a}}\nfrac{1}{\delta})$[\Cref{thm:spc_mn_test}]} \\\cline{1-1}\cline{3-3}
   
   \makecell{$\eps_v n$ arbitrary preferences\\ and $\eps_a m$ alternatives away} &  & \makecell{$\OO(\tfrac{1}{(1-3\eps_v)^2}\log \tfrac{\log_{\nfrac{1}{\eps_a}}\nfrac{1}{\delta}}{\delta} \log_{\nfrac{1}{\eps_a}}\tfrac{1}{\delta}\log\log_{\nfrac{1}{\eps_a}}\nfrac{1}{\delta})$\\ for $\eps_v<\nfrac{1}{3}^\star$ [\Cref{cor:spc_mn_test_worst}]} \\\hline}
   
   \makecell{$\eps_v n$ arbitrary preferences away} & \makecell{$\eps_v^\pr n$ arbitrary\\ preferences away} & \makecell{$\OO(\tfrac{1}{(\eps_v^\pr-\eps_v)^2}(2^m m^2\log^2 m + \log \nfrac{1}{\delta}))$ [\Cref{cor:spc_n_test_worst_worst}]} \\\hline
   
   \longversion{
   $\eps_a m$ alternatives away & $\eps_a^\pr m$ alternatives away & \multirow{2}{*}{\makecell{$\Omega(n\log \nfrac{1}{\delta})$ even for $\eps_a=0$ and\\ for every $0<\eps_a^\pr\le 1$ and $0<\delta<\nfrac{1}{2}$ [\Cref{thm:sp_m_test_worst_worst}]}} \\\cline{1-2}
   
   \makecell{$\eps_v n$ arbitrary preferences\\ and $\eps_a m$ alternatives away} & \makecell{$\eps_v^\pr n$ arbitrary preferences\\ and $\eps_a^\pr m$ alternatives away} &\\\hline
}
   
   \shortversion{
   $\eps_a m$ alternatives away & $\eps_a^\pr m$ alternatives away & \makecell{$\Omega(n\log \nfrac{1}{\delta})$ even for $\eps_a=0$ and\\ for every $0<\eps_a^\pr\le 1$ and $0<\delta<\nfrac{1}{2}$ [\Cref{thm:sp_m_test_worst_worst}]} \\\hline
   }
  \end{tabular}
 \end{center}
  \caption{Summary of results for distinguishing profiles in the first column from the profiles in the second column; all the distances are from the single peaked domain. \longversion{By saying that a profile is $\eps n$ random preferences away, we mean that we need to delete $\eps n$ preferences to make the resulting profile single peaked; moreover the deleted preferences are distributed uniformly at random in the set of all preferences with replacement. Similarly, we say a profile to be $\eps n$ arbitrary preferences away if we need to delete $\eps n$ preferences to make the resulting profile single peaked (with no assumption about them). We say a profile is $\eps_v n$ arbitrary preferences and $\eps_a m$ alternatives away if we need to simultaneously delete $\eps_v n$ preferences and $\eps_a m$ alternatives to make the resulting profile single peaked. We assume $0\le\eps_v < \eps_v^\pr\le1$ and $0\le\eps_a<\eps_a^\pr\le1$. }$\star:$ For any $0\le\eps_v<1$, we refer to \Cref{cor:gen_n_test_worst_any_eps}\longversion{ and \Cref{cor:gen_mn_test_worst}}. Refer to \Cref{subsec:model} for our sampling model.}\label{tbl:summary}
\end{table*}

Learning users' preferences is useful in the contexts of social choice, recommender systems, product development, and many more applications. It is often observed that preferences are never completely arbitrary, rather they possess correlated structures~\cite{gaertner2001domain}. For example, preferences of citizens for a facility location have a {\em single peaked} structure~\cite[Section 1]{FilosRatsikas}, i.e., a citizen has highest preference for the facility at her location and it monotonically decreases with the distance from her. This kind of preferences are also prevalent in political opinions based on the voters' bias to the conservative or liberal views~\cite{hinich1997analytical}. Intuitively, in a single peaked preference profile, we assume that there exists a societal axis where the alternatives have been ordered and every preference ``respects'' that ordering in the following sense. Every preference has an implicit most preferred point $t$ on the societal axis and if an alternative $x$ lies between $t$ and another alternative $y$, then $x$ is preferred over $y$. The advantage of preferences with such structures is that they can efficiently bypass the classic impossibility results of social choice theory~\cite{arrow1950difficulty,gibbard1973manipulation,satterthwaite1975strategy}. 

Similarly, in the design of recommender systems, it has often been observed that users' preferences (and hence their recommendations) have patterns that are (a) demography-based, (b) knowledge-based, (c) feature-based, or (d) content based~\cite{PennockHG00}. While designing a product, an enterprise may wish to look for structures in the end users' preferences, and design their product such that a collectively `efficient' choice is made to cater a large number of users.

While it is difficult to predict the users' preferences apriori, data on the preferences, obtained through users' purchase and browsing patterns, or through surveys, are plentiful which are classified into demography, knowledge, affinity towards a feature or content. It remains to discover whether the preferences come from a specific class that we call {\em preferential domains} or simply, {\em domains}. 

A domain is a collection of sets of preferences over a set of alternatives. A preference profile, i.e., the tuple of preferences of all the agents/users, is said to belong to a domain if, for some set in the domain, every preference in the profile belongs to that set. 
\begin{example}[Single peaked domain]
 Consider three alternatives $a,b,c$. The {\em single peaked domain} with these alternatives is denoted by $\DD = \{\BB_1,\ldots, \BB_6\}$, where $\BB_1=\{(abc),(bac),(bca),(cba)\}$ when the societal order over the alternatives is $a \prec b \prec c$, and similarly, $\BB_2, \ldots, \BB_6$ are the sets of preferences over the same alternatives for different societal orders of $a,b$, and $c$.
\end{example}
Some prominent examples of domains are {\em single peaked, single crossing, Euclidean,}~\cite{gaertner2001domain} etc. The benefit of the discovery of such domains (even as a partial population) is that a much refined plan or protocol can be designed for such domains which satisfy several desirable axioms. For example, the {\em median voting rule} in the single peaked domain ensures that no voter can gain by misreporting her preference~\cite{moulin1991axioms}. 
Another reason to study various domains concerns computational considerations. Indeed, some of the most fundamental problems in computational social choice, for example, computing winners for many important voting rules such as Kemeny\longversion{ \cite{kemeny1959mathematics,levenglick1975fair}}, Dodgson\longversion{ \cite{dodgson1876method,black1958theory}}, and Young\longversion{ \cite{young1977extending}} are computationally intractable \longversion{\cite{bartholdi1989voting,hemaspaandra2005complexity,procaccia2008complexity}}\shortversion{\cite{brandt2016handbook}}. It turns out that most of these problems become efficiently solvable in many domains, single peaked for example~\cite{brandt2015bypassing}.


Our work in this paper contributes to uncovering whether a given preference profile is ``close'' to some domain, through sampling a small number of preferences and/or alternatives. The guarantees we provide are probabilistic that converges to unity as more preferences/alternatives are investigated -- the cost of such an investigation is often proportional to the number of samples drawn, known as {\em sample complexity}. Hence our goal is to minimize the sample complexity of our algorithms. For example, our algorithms could be used to predict whether there exist at least, say 95\%, of the preferences in a profile which are single peaked. If we know the societal order of the single peaked preferences (which constitute at least 95\% of the profile), using median voting rule on the single peaked sub-profile would yield all the desirable properties of the median voting rule, e.g., truthfulness for those 95\% of the population. These kind of truthfulness of a fraction of voters is referred to as ``approximate truthfulness.''
In many applications like public good provisioning, it is highly beneficial to uncover truthful opinions from the vast majority of the population.

To put our work in perspective, we revisit a question that is often asked in computational social choice for any domain. This is about the existence of an efficient recognition algorithm: given a profile \PP, does there exist a polynomial time algorithm to decide whether \PP belongs to the domain? There exist efficient recognition algorithms for many popular domains, for example, single peaked~\cite{bartholdi1986stable}, single crossing~\cite{DoignonF94}, etc.~
\cite{knoblauch2010recognizing,ElkindF14,ElkindFLO15,MagieraF17}. One notable exception is the Euclidean domain of dimension\longversion{ at least} two where the recognition problem is \NPH~\cite{Peters17}.

There are two main limitations of the recognition problem. First, the problem formulation is ``exact.'' Real world profiles are almost never perfect and thus they can only be at most ``close'' to some domain. More specifically, there may be few preferences \longversion{and/or}\shortversion{or} alternatives (treated as outliers) whom we need to ignore to obtain the required structure.  Unfortunately, outliers' consideration often makes the related recognition problem intractable, (e.g., the voter deletion for single peaked domain~\cite{ErdelyiLP17}). Second, the recognition problem needs access to the entire preference profile. In many situations, e.g., pre-election polls, surveys, etc., we only have access to samples. In other cases, the number of preferences may be too large and, depending on the application at hand, a sub-linear time (possibly approximation) algorithm may be more useful. We address both these issues by defining a related testing problem. As a concrete use case, a social planner could use our testing algorithms to know whether it is possible to remove, say 5\% of the preferences to obtain a single peaked structure by observing a small number of samples.

A corresponding computational problem is: can a profile of $n$ preferences over $m$ alternatives belong to some domain \DD after deleting, say at most $k$ preferences (or alternatives), by drawing a small number of samples? However, any algorithm for this problem would\longversion{ provably} need to observe $\Omega(n)$ samples which defeats the main purpose of testing (except when \DD is empty or \DD contains all possible profiles). To see this, let us consider a specific case of \DD to be single peaked; the set of alternatives be $\{a,b,c\}$. Let \PP be a profile consisting of $\nfrac{n}{2}$ (say $n$ is an even integer) copies of $a\suc b\suc c$, $(\nfrac{n}{2})-1$ copies of $a\suc c\suc b$, and one $c\suc b\suc a$. We observe that \PP is not single peaked after observing the last preference $c\suc b\suc a$. However, deletion of that preference makes it single peaked. Let us now consider another profile \QQ consisting of $\nfrac{n}{2}$ copies of $a\suc b\suc c$, $(\nfrac{n}{2})-2$ copies of $a\suc c\suc b$, and two copies of $c\suc b\suc a$. Again, \QQ is not single peaked, but deletion of the two copies of $c\suc b\suc a$ makes it single peaked. We now observe that the KL-divergence~\cite{kullback1951information} between the two distributions of samples for \PP and \QQ is $\OO(\nfrac{1}{n})$ and thus distinguishing \PP from \QQ (which any testing algorithm has to do) requires $\Omega(n)$ samples to succeed with any constant nonzero probability~\cite{BarYossef03}. To overcome this lower bound, we introduce (as is ubiquitous in testing literature~\cite{ron2001property,goldreich1999combinatorial}) a ``gap'' in the two possible inputs. In all our testing problems, we are given a profile\longversion{ of $n$ preferences over $m$ alternatives} as input which is guaranteed to be one of the two possible types, and we need to find which one it is. The two possibilities for the input will cover all the cases except few and thus there is a ``gap.''

\subsection{Our Contribution}

Our specific contribution in this paper are as follows. The error probability of any algorithm below is at most $\delta\in(0,1)$.
\begin{enumerate}[(i),leftmargin=1pt,itemindent=*]
 \item We present a sampling based algorithm to distinguish any profile for which there exists a set \RR of at most $\eps_v n$ preferences (or $\eps_a m$ alternatives) whose deletion makes the resulting profile belong to \DD from any random profile (refer to the first three rows in \Cref{tbl:summary}). We observe that the sample complexity depends on whether we assume \RR to be arbitrary or random. We remark that, in the testing literature~\cite{GoldreichGR98,andrews1998tests,shao2011testing}, it is popular to assume the noise to be random which is equivalent to assuming the preferences in \RR to be random in our context.
 
 \item For any $0\le\eps_v < \eps_v^\pr\le1$, we present a sampling based algorithm to distinguish any profile for which there exist at most $\eps_v n$ preferences whose deletion makes the resulting profile belong to \DD from any profile where one has to delete at least $\eps_v^\pr n$ preferences to make it belong to \DD (refer to the fourth row in \Cref{tbl:summary}).
 
 \item In the case of alternatives, we prove that any algorithm for distinguishing any profile for which there exist at most $\eps_a m$ alternatives whose deletion makes the resulting profile belong to \DD from any profile where one has to delete at least $\eps_a^\pr m$ alternatives to make it belong to \DD has sample complexity of $\Omega(n\log \nfrac{1}{\delta})$ for every $0\le\eps_a < \eps_a^\pr\le1$ even when $\eps_a=0$ (refer to the fifth row in \Cref{tbl:summary}). This shows that detecting arbitrary {\em outlier alternatives} is much harder than detecting arbitrary {\em outlier preferences} from a sample complexity viewpoint.
\end{enumerate}

We remark that all our results in \Cref{tbl:summary} for the single peaked domain actually extend to any domain as described in \Cref{sec:results}. From a technical point of view, to tackle preferences which are outliers, we define and exploit a notion called {\em content of a domain} which, informally, is the maximum number of distinct preferences that any profile in the domain can contain as a function of the number of alternatives. On the other hand, we blend with it the ideas from the classical coupon collector problem to handle alternatives which are outliers. To develop an algorithm for the case when the outliers can be arbitrary, we prove a key structural result (in \Cref{lem:worst}) for arbitrary domain which may be of independent interest also.

\subsection{Related Work}

The computational problem of recognizing whether a given profile belongs to a domain\longversion{ (called the recognition problem)} has been studied extensively in computational social choice. Trick~\longversion{\cite{bartholdi1986stable}}\shortversion{\shortcite{bartholdi1986stable}} shows that the recognition problem is polynomial time solvable for single peaked profiles. Escoffier et al.~\longversion{\cite{EscoffierLO08}}\shortversion{\shortcite{EscoffierLO08}} improve the efficiency of the recognition algorithm for the single peaked profiles. Elkind et al.~\longversion{\cite{ElkindFS12}}\shortversion{\shortcite{ElkindFS12}} present a polynomial time algorithm for recognizing single crossing profiles. Barber{\`{a}} and Moreno~\longversion{\cite{barbera2011top}}\shortversion{\shortcite{barbera2011top}} discover a property called top monotonicity which simultaneously generalizes both single peakedness and single crossingness. Magiera and Faliszewski~\longversion{\cite{MagieraF17}}\shortversion{\shortcite{MagieraF17}} present polynomial time recognition algorithm for top monotonic profiles. Doignon and Falmagne~\longversion{\cite{DoignonF94}}\shortversion{\shortcite{DoignonF94}} show that the recognition problem for the one dimensional Euclidean domain is polynomial time solvable. Knoblauch~\longversion{\cite{knoblauch2010recognizing}}\shortversion{\shortcite{knoblauch2010recognizing}} and Elkind and Faliszewski~\longversion{\cite{ElkindF14}}\shortversion{\shortcite{ElkindF14}} present alternative algorithms for recognizing one dimensional Euclidean profiles. Peters~\longversion{\cite{Peters17}}\shortversion{\shortcite{Peters17}} shows that recognizing Euclidean profiles of dimension at least two is \NPH.

Lackner~\longversion{\cite{Lackner14}}\shortversion{\shortcite{Lackner14}} shows that the computational problem of finding if it is possible to extend a given incomplete profile to a single peaked profile is \NPC. However, if we restrict ourselves to only weak orders, then the computational problem of recognizing incomplete single peaked profiles is polynomial time solvable~\cite{Fitzsimmons14}. The above problem is polynomial time solvable for single crossing profiles too~\cite{ElkindFLO15}. Erd{\'{e}}lyi~\longversion{\cite{ErdelyiLP17}}\shortversion{\shortcite{ErdelyiLP17}} studies complexity of the computational problem of deciding whether a given profile can be ``made'' single peaked by deleting few preferences or alternatives; Bredereck et al.~\longversion{\cite{BredereckCW16}}\shortversion{\shortcite{BredereckCW16}} study complexity of this problem for single peaked, single-caved, single-crossing, \longversion{value-restricted, best-restricted, worst-restricted, medium-restricted, group-separable,} etc. profiles.
Ballester and Haeringer~\longversion{\cite{ballester2011characterization}}\shortversion{\shortcite{ballester2011characterization}} present characterization of  single peaked profiles through succinct forbidden configurations. Bredereck et al.~\longversion{\cite{bredereck2013characterization}}\shortversion{\shortcite{bredereck2013characterization}} show forbidden configurations for the single crossing profiles. Elkind et al.~\longversion{\cite{ElkindFS14}}\shortversion{\shortcite{ElkindFS14}} present forbidden configurations for profiles which are simultaneously single peaked and single crossing. A related literature studies the likelihood of a random profile being single peaked~\cite{lackner2017likelihood,DBLP:journals/dm/ChenF18,chatterji2016characterization}.


%% file: prelim.tex
\section{Preliminaries and Problem Formulation}\label{sec:prelim}

For any two positive integers $k$ and \el with $k\le\el$, we denote the set $\{j\in\NB:1\le j\le k\}$ by $[k]$ and the set $\{j\in\NB: k\le j\le\el\}$ by $[k,\el]$. For a set \XX, we denote its power set by $2^\XX$. Let \AA be a finite set of alternatives of cardinality $m$. Preferences are linear orders over \AA. We denote the set of all linear orders over \AA by $\LL(\AA)$. For any positive integer $n$, a tuple $\left(\suc_i\right)_{i\in[n]}$ of $n$ preferences is called a profile. If not mentioned otherwise, we use $m,n,$ and \AA to denote the number of alternatives, the number of preferences in a profile, and the set of alternatives, respectively. For a subset $\XX\subseteq\AA$ and a preference ${\suc}\in\LL(\AA)$, we denote the restriction of \suc to \XX by $\suc(\XX)$. A  preferential domain or simply domain is a collection of subsets of $\cup_{|\AA|>0}\LL(\AA)$. We call a domain \DD nontrivial if $\DD\ne\emptyset$ and $\DD\ne\cup_{|\AA|>0}2^{\LL(\AA)}$. Given a domain \DD and a profile $\PP=\left(\suc_i\right)_{i\in[n]}$ over \AA, we say (with slight abuse of notation) that $\PP\in\DD$ if there exists a $\BB\in\DD$ such that $\PP\in\BB^n$. We call a domain \DD \emph{neutral} if whenever $\left(\suc_i\right)_{i\in[n]}\in\DD$, we have $(\sigma\left(\suc_i\right))_{i\in[n]}\in\DD$ for every permutation $\sigma$ of $[m]$; if $\suc_i$ is defined as $a_1\suc a_2\suc \cdots\suc a_m$, then $\sigma\left(\suc_i\right)$ is defined as $a_{\sigma(1)}\suc a_{\sigma(2)}\suc \cdots\suc a_{\sigma(m)}$. We call a domain \DD \emph{normal} if whenever $\left(\suc_i\right)_{i\in[n]}\in\DD$, we have $(\suc_i(\XX))_{i\in[n]}\in\DD$ for every $\XX\subseteq\AA$. In this work, we consider only neutral and normal domains. We remark that many popular domains including single peaked, single caved, single crossing, top restricted, bottom restricted, etc. satisfy these two properties (the only notable exception is the domain of top monotonic~\cite{barbera2011top} profiles).

Let \DD be any domain and $\PP = \left(\suc_i\right)_{i\in[n]}\in\LL(\AA)^n$ be a profile. If it satisfies the following conditions:
 \begin{enumerate}[label=(\roman*),leftmargin=*]
  \item there exists a subset $J\subset[n]$ such that $|J|=\el$ and $\left(\suc_i\right)_{i\in J}\in\DD$, and
  \item for every subset $K\subset[n]$ such that $|K|>\el$, we have $\left(\suc_i\right)_{i\in K}\notin\DD$,
 \end{enumerate}
then we say that the {\em preference-distance} of \PP from \DD is $(n-\el)$, and we call the preferences which need to be deleted to bring the profile back to \DD to be {\em preference outliers}. Similarly, we can define the notion of {\em alternative-distance} (where only alternatives need to be deleted) and alternative outliers.

Our first problem is to distinguish a profile which is, informally speaking, $\eps_a m$ alternatives and random $\eps_v n$ preferences away from some domain \DD vs a random profile. We call this problem \randomtest{$\eps_v$}{$\eps_a$}{$\delta$} which is formally defined as follows.

\vspace{1ex}
 \noindent\fbox{\begin{minipage}{\linewidth}
\begin{probdef}[\randomtest{$\eps_v$}{$\eps_a$}{$\delta$}]\label{def:rand_rand}
 Let $\left(\suc_i\right)_{i\in[n]}$ be a profile over a set \AA of alternatives which is either one of the following kind:
 \begin{enumerate}[label=(\roman*),leftmargin=*]
  \item There exists $\WW\subseteq[n]$ and $\XX\subseteq\AA$ with $|\WW|\ge (1-\eps_v)n$ and $|\XX|\ge(1-\eps_a)m$ such that the profile $(\suc_i(\XX))_{i\in\WW}$ belongs to the domain \DD and $\suc_j(\XX)$ is distributed uniformly in $\LL(\XX)$ for every $j\in[n]\setminus\WW$.
  \item The preference $\suc_i$ is distributed uniformly randomly in $\LL(\AA)$ for every $i\in[n]$.
 \end{enumerate}
 Output $1$ if the input profile is of the first kind and $0$ if it is of the second kind; the probability of error can be at most $\delta$.\longversion{ The probability is taken over the randomness used in generating the instances and the randomness used by the algorithm.}
\end{probdef}
\end{minipage}}
\vspace{1ex}

\Cref{def:rand_rand} assumes that the preference outliers are distributed uniformly randomly which can be a strong assumption depending on the application at hand. The \worstcasetest{$\eps_v$}{$\eps_a$}{$\delta$} problem in \Cref{def:worst_rand} removes this assumption.

\vspace{1ex}
 \noindent\fbox{\begin{minipage}{\linewidth}
\begin{probdef}[\worstcasetest{$\eps_v$}{$\eps_a$}{$\delta$}]\label{def:worst_rand}
 Let $\left(\suc_i\right)_{i\in[n]}$ be a profile over a set \AA of alternatives which is either one of the following:
 \begin{enumerate}[label=(\roman*),leftmargin=*]
  \item There exists $\WW\subseteq[n]$ and $\XX\subseteq\AA$ with $|\WW|\ge (1-\eps_v)n$ and $|\XX|\ge(1-\eps_a)m$ such that the profile $(\suc_i(\XX))_{i\in\WW}$ belongs to the domain \DD.
  \item The preference $\suc_i$ is distributed uniformly randomly in $\LL(\AA)$ for every $i\in[n]$.
 \end{enumerate}
 Output $1$ if the input profile is of the first kind and $0$ if it is of the second kind; the probability of error can be at most $\delta$.\longversion{ The probability is taken over the randomness used in generating the instances in (ii) and the randomness used by the algorithm.}
\end{probdef}
\end{minipage}}
\vspace{1ex}

\Cref{def:worst_rand} still retains the assumption from \Cref{def:rand_rand} that the second possibility for the input profile is random. The \worstworstcasetest{$\eps_v$}{$\eps_a$}{$\eps_v^\pr$}{$\eps_a^\pr$}{$\delta$} problem in \Cref{def:worst_worst} is the most general problem in our paper which removes all these structural assumptions from \Cref{def:rand_rand,def:worst_rand}.

\vspace{1ex}
 \noindent\fbox{\begin{minipage}{\linewidth}
\begin{probdef}[\worstworstcasetest{$\eps_v$}{$\eps_a$}{$\eps_v^\pr$}{$\eps_a^\pr$}{$\delta$}]\label{def:worst_worst}
 Let $\left(\suc_i\right)_{i\in[n]}$ be a profile over a set \AA of alternatives which is either one of the following kind where $0\le\eps_v<\eps_v^\pr\le1$ and $0\le\eps_a<\eps_a^\pr\le1$:
 \begin{enumerate}[label=(\roman*),leftmargin=*]
  \item There exists $\WW\subseteq[n]$ and $\XX\subseteq\AA$ with $|\WW|\ge (1-\eps_v)n$ and $|\XX|\ge(1-\eps_a)m$ such that the profile $(\suc_i(\XX))_{i\in\WW}$ belongs to the domain \DD.
  \item For every $\WW\subseteq[n]$ and $\XX\subseteq\AA$ with $|\WW|> (1-\eps_v^\pr)n$ and $|\XX|>(1-\eps_a^\pr)m$, the profile $(\suc_i(\XX))_{i\in\WW}$ does not belong to the domain \DD.
 \end{enumerate}
 Output $1$ if the input profile is of the first kind and $0$ if it is of the second kind; the probability of error can be at most $\delta$.\longversion{ The probability is taken over the randomness used by the algorithm.}
\end{probdef}
\end{minipage}}
\vspace{1ex}

\shortversion{In \Cref{def:rand_rand,def:worst_rand,def:worst_worst}, the error probability is taken over the randomness used in generating the instances in (ii) and the randomness used by the algorithm.}

\subsection{Content and Residue of Domain}

We now define the {\em content} and {\em residue} of any domain which will make the many of our results simpler to state. Let \DD be any domain. We define the {\em content} of \DD as a function $\ccc_\DD: \NB\longrightarrow[0,1]$ such that any profile with $m\in\NB$ alternatives in \DD can have at most $\ccc_\DD(m) m!$ distinct preferences; we call the function $\rrr_\DD: \NB\longrightarrow[0,1]$ defined as $\rrr_\DD(m) = 1-\ccc_\DD(m)$ the {\em residue} of a domain. For example, $\ccc_{\text{single peaked}}(2)=1,\ccc_{\text{single peaked}}(3)=\nfrac{2}{3}$, $\ccc_{\text{single crossing}}(m) = \nfrac{\left({m\choose 2}+1\right)}{m!}$~\cite{DeyM16a}. For technical reason, let us assume that $\ccc_\DD(1)=1$ for every \DD. We observe that, for normal domains, the function $\ccc_\DD(\cdot)$ is non-increasing (and thus $\rrr_\DD(\cdot)$ is a non-decreasing function). Whenever the domain \DD is immediate from the context, we omit \DD from subscript of $\ccc$ and \rrr.


\subsection{Sampling Model and Sample Complexity}\label{subsec:model}

In our model, there is an oracle which, when queried, returns an agent $v$ picked uniformly randomly with replacement from the set of all agents. Now the algorithm can ask the agent $v$ an arbitrary number of comparison queries -- in a comparison query, two alternatives $x$ and $y$ are presented to the agent $v$ and it replies whether it prefers $x$ over $y$ or $y$ over $x$. The {\em sample complexity} of an algorithm is defined to be the total number of comparison queries it makes during its execution. We remark that defining sample complexity (instead of the number of agents sampled) as the number of comparison queries enables us to perform more fine grained analysis of the complexity of our problems.

\subsection{Chernoff Bound}

We repeatedly use the following concentration inequality:
\begin{theorem}\label{thm:chernoff}
Let $X_1, \dots, X_\ell$ be a sequence of $\ell$ independent
random variables in $[0,1]$ (not necessarily identical). Let $S = \sum_i X_i$ and
let $\mu = \E{S}$. Then, for any $0 \leq \delta \leq 1$: 
\begin{equation}
 \label{eq:additive}
 \Pr[|S-\mu| \geq \delta \ell] < 2 \exp(-2\ell \delta^2),
\end{equation}
and 
\begin{equation}
 \label{eq:multiplicative}
 \Pr[|S - \mu| \geq \delta \mu] < 2\exp(-\delta^2\mu/3).
\end{equation}
\Cref{eq:additive,eq:multiplicative} are called additive and multiplicative versions of the bound respectively. 
\end{theorem}

%% file: results.tex
\section{Results}
\label{sec:results}

We now present our main results. Our general approach would be to explain our algorithms for the special case of the single peaked domain first and then generalize to arbitrary domain; we make an exception for few cases where presenting the general case directly better reveals the key idea.\shortversion{ In the interest of space, we omit some of our proofs, which can be found in the supplemental material. For ease of exposition and interest of space, we have deferred our more involved algorithms for the cases when both preferences and alternatives could simultaneously be outliers to the supplemental material.}

\subsection{Only Preferences as Outliers}

In this subsection, we focus on the case when only preferences are considered as outliers. We begin with presenting our \randomtester{$\eps_v$}{$0$}{$\delta$} for the single peaked domain.\longversion{ We then generalize it to arbitrary domains in \Cref{cor:gen_n_test}.} Our algorithm first fixes any three alternatives, say $a, b,$ and $c$. Then it samples few preferences restricted to these three alternatives only. If all the six possible permutations of $a, b,$ and $c$ appear nearly same number of times, then the algorithm predicts the profile to be a random profile; otherwise it predicts it to be close to single peaked. We now formally present our algorithm in \Cref{thm:spc_n_test}.

\begin{algorithm}[!htbp]
 \caption{\randomtester[single peak]{$\eps_v$}{$0$}{$\delta$}\label{code:pref_outlier}}
  \begin{algorithmic}[1]
    \Require{Oracle access to a profile \PP}
    \Ensure{$1$ if there exists $\eps_v n$ preferences whose deletion makes the resulting profile single peaked and $0$ if \PP has been generated randomly}
    \State{Let $a,b,c\in\AA$ be any $3$ arbitrary alternatives}
    \State{Sample $\el=\tfrac{72}{(1-\eps_v)^2}\ln\tfrac{6}{\delta}$ preferences restricted to $\{a,b,c\}$ uniformly at random from the input profile with replacement. Let $\BB\in\LL(\{a,b,c\})^\el$ be the profile of sampled preferences}
    \State{Let $t$ be the minimum number of times any preference in $\LL(\{a,b,c\})$ appear in \BB}
    \If{$t<\tfrac{\el}{12}(1+\eps_v)$}
        \Return $1$
    \Else
        \Return $0$
    \EndIf
  \end{algorithmic}
\end{algorithm}

\begin{theorem}\label{thm:spc_n_test}
 For at least $3$ alternatives, there exists a \randomtester[single peak]{$\eps_v$}{$0$}{$\delta$} with sample complexity $\OO(\tfrac{1}{(1-\eps_v)^2}\log\tfrac{1}{\delta})$ for every $0\le\eps_v<1$ and $0<\delta<\nfrac{1}{2}$. If there are only $2$ alternatives, then there does not exist any such tester.
\end{theorem}

\begin{proof}
 For $m=2$, the result follows from the observation that a profile where every preference is distributed uniformly in the set of all possible preferences is single peaked and thus the two cases are statistically indistinguishable. So let us assume $m\ge 3$ and $a, b,$ and $c$ be any three alternatives. \longversion{We present our algorithm in \Cref{code:pref_outlier}. }We pick $\el=\tfrac{72}{(1-\eps_v)^2}\ln\tfrac{6}{\delta}$ preferences uniformly at random with replacement and query oracle to know how $a, b,$ and $c$ are ordered in these preferences. Let $p_i, i\in[6]$, be all possible permutations of $\{a,b,c\}$ and $X_i$ be the random variable denoting the number of sampled preferences where the permutation $p_i$ appears for $i\in[6]$. We output $1$ if $\min_{i\in[6]} X_i < \tfrac{\el}{12}(1+\eps_v)$ and output $0$ otherwise. We observe that the sample complexity of our algorithm is $6\el=\OO(\tfrac{1}{(1-\eps_v)^2}\ln\tfrac{1}{\delta})$. We now turn to the correctness of our algorithm. For that we show that irrespective of the input profile, the probability of making an error is at most $\delta$.
 
 \begin{itemize}[leftmargin=0cm,itemindent=0.3cm,labelwidth=\itemindent,labelsep=0cm,align=left,noitemsep,topsep=2pt]
 
  \item {\bf Case I - the input profile is single peaked after deleting at most $\eps_v n$ preferences which are distributed uniformly:} Let \PP be the input profile and \QQ be a sub-profile of \PP which is single peaked and contains at least $(1-\eps_v)n$ preferences. Hence, there exists an $\eta\in[6]$ such that the preference $p_\eta$ does not appear in \QQ. Since the preferences in $\PP\setminus\QQ$ is uniformly distributed and $|\PP\setminus\QQ|\le\nfrac{\eps_v n}{6}$, we have $\EB[X_\eta]\le\nfrac{\eps_v\el}{6}$. Using Chernoff bound (additive form), we now have the following:
  \[ \Pr[\text{error}]\le \Pr[X_j \ge \tfrac{\el}{12}(1+\eps_v)] \le \exp\{-\tfrac{\el(1-\eps_v)^2}{72}\} \le \delta \]
  
  \item {\bf Case II - the input profile is distributed uniformly:} Since every preference in profile \PP is uniformly distributed, for every $i\in[6]$, we have $\EB[X_i]=\nfrac{\el}{6}$. Using Chernoff bound (multiplicative form) followed by union bound, we have the following:
  \begin{align*}
   \Pr[\text{error}] &= \Pr[\exists i\in[6], X_i\le \tfrac{\el}{12}(1+\eps_v)]\\
   &\le 6\exp\{-\nfrac{(1-\eps_v)^2 \el}{48}\}\le \delta\qedhere
  \end{align*}
 \end{itemize}
\end{proof}

The main idea in \Cref{thm:spc_n_test} can be easily extended to arbitrary domains.

\begin{corollary}\label{cor:gen_n_test}
 Let \DD be any normal and neutral domain and $m_0=\min\{m: \ccc_\DD(m)<1\}$. For at least $m_0$ alternatives, there exists a \randomtester{$\eps_v$}{$0$}{$\delta$} with sample complexity $\OO(\tfrac{1}{(1-\eps_v)^2}\ln\tfrac{1}{\delta})$ for every $0\le\eps_v<1$ and $0<\delta<\nfrac{1}{2}$. If the number of alternatives is at most $m_0-1$, then there does not exist any such tester.
\end{corollary}

\longversion{
\begin{proof}
 The argument for the impossibility result is exactly same as in the proof of \Cref{thm:spc_n_test}. So let us assume that the number of alternatives $m$ is at least $m_0$. We pick a subset $\BB\subset\AA$ of $m_0$ alternatives arbitrarily where \AA is the set of alternatives in the input profile. We now sample $\el=\OO(\tfrac{1}{(1-\eps_v)^2}\ln\tfrac{1}{\delta})$ as in \Cref{thm:spc_n_test} (although the constant hidden in \OO is different from \Cref{thm:spc_n_test} and depends on $m_0$), query oracle to learn these preferences restricted to \BB. We output $1$ if there exists a preference ${\suc}\in\LL(\BB)$ which is observed in the sampled preferences at most $\tfrac{\el}{2(m_0!)}(\eps_v+1)$ and output $0$ otherwise. The sample complexity of our algorithm is $\OO(\tfrac{1}{(1-\eps_v)^2}\ln\tfrac{1}{\delta})$ (since $m_0$ is a constant). The proof of correctness of our algorithm is along the same lines as in the proof of \Cref{thm:spc_n_test}.
\end{proof}
}

We now turn our attention to the \worstcasetest{$\eps_v$}{$0$}{$\delta$} problem; that is when the outliers can be arbitrary (need not be randomly generated). We begin with presenting a general impossibility result in this case. Its proof follows from the observation that, in this case, one can carefully construct the set of outliers so that the distribution of samples in both the possibilities are statistically indistinguishable.

\begin{proposition}\label{prop:m_lb}
 For every domain \DD, there does not exist any \worstcasetester{$\eps_v$}{$0$}{$\delta$} for any $\eps_v\ge\rrr_\DD(m)$ where $m$ is the number of alternatives in the input profile.
\end{proposition}

\longversion{
\begin{proof}
 Let \AA be a set of alternatives with $|\AA|=m$. Let $\BB\subseteq\LL(\AA)$ be a set of distinct preferences such that $|\BB|=(1-\rrr_\DD(m))m!$ and \BB belongs to the domain \DD; there exists such a \BB from the definition of the residue of domain. For any $n$ divisible by $\nfrac{1}{\rrr_\DD(m)}$, let us consider the profile \PP consisting of $\nfrac{n(1-\rrr_\DD(m))}{|\BB|}=\nfrac{n}{m!}$ copies of every preference in \BB and $\nfrac{n\rrr_\DD(m)}{(m!-|\BB|)}=\nfrac{n}{m!}$ copies of every preference in $\LL(\AA)\setminus\BB$. We observe that the preference-distance of \PP from \DD is at most $\rrr_\DD(m) n\le \eps_v n$. Let \QQ be another profile where every preference is picked uniformly at random with replacement from the set of all $m!$ preferences. Now the result follows from the observation that the distribution of a preference picked uniformly at random with replacement from \PP is the same as the corresponding distribution for \QQ.
\end{proof}
}

We now present our \worstcasetester[single peak]{$\eps_v$}{$0$}{$\delta$} for $\eps_v<\nfrac{1}{3}$ in \Cref{thm:spc_n_test_worst}. We defer our general \worstcasetester{$\eps_v$}{$0$}{$\delta$} till \Cref{cor:gen_n_test_worst_any_eps} which not only handles every $\eps_v<1$ but also takes care of arbitrary domain (but the sample complexity will be worse than that of \Cref{thm:spc_n_test_worst}). The main idea of the algorithm in \Cref{thm:spc_n_test_worst} is exactly the same as the algorithm in \Cref{thm:spc_n_test} -- it samples some preferences restricted to any $3$ alternatives and outputs that the profile is random if all the $6$ possible permutations appear nearly equal number of times; otherwise it says that the profile is close to single peaked.

\begin{theorem}\label{thm:spc_n_test_worst}
 There exists a \worstcasetester[single peak]{$\eps_v$}{$0$}{$\delta$} with sample complexity $\OO(\tfrac{1}{(1-3\eps_v)^2}\ln\tfrac{1}{\delta})$ for every $0\le\eps_v<\nfrac{1}{3}$.
\end{theorem}
\begin{proof}
 \longversion{Our algorithm for this case is similar to \Cref{thm:spc_n_test} and thus we only elaborate the differences.} As in \Cref{thm:spc_n_test}, we choose any $3$ alternatives $a, b,$ and $c$, pick $\el=\tfrac{72}{(1-3\eps_v)^2}\ln\tfrac{6}{\delta}$ preferences uniformly at random with replacement, and query oracle to know how $a, b,$ and $c$ are ordered in these preferences. We output $1$ if $\min_{i\in[6]} X_i < \tfrac{\el}{12}(1+3\eps_v)$ and output $0$ otherwise (with notation as defined in the proof of \Cref{thm:spc_n_test}). The proof of correctness and the analysis of the sample complexity of our algorithm is similar to \Cref{thm:spc_n_test} using the observation that, when the input profile can be made single peaked by deleting at most $\eps_v n$ preferences, there exists an $\eta\in[6]$ such that $X_\eta\le\nfrac{\eps_v\el}{2}$ since $\rrr_{\text{single peak}}(3)=\nfrac{1}{3}$.
\end{proof}

From the proof of \Cref{thm:spc_n_test_worst}, the following generalization to arbitrary domain is immediate.

\begin{corollary}\label{cor:gen_n_test_worst}
 Let $m_0=\min\{m\in\NB: \rrr_\DD(m)<1\}$ and the number of alternatives is at least $m_0$. Then there exists a \worstcasetester{$\eps_v$}{$0$}{$\delta$} with sample complexity $\OO(\tfrac{1}{(1-(\nfrac{\eps_v}{\rrr_\DD(m_0)}))^2}\ln\tfrac{1}{\delta})$ for every $\eps_v$ with $0\le\eps_v<\rrr_\DD(m_0)$ (the \OO notation in the sample complexity hides constant which depends on $m_0$).
\end{corollary}

We now present our \worstcasetester{$\eps_v$}{$0$}{$\delta$} for any $\eps_v<1$ generalizing \Cref{thm:spc_n_test_worst}. Of course we need the number of alternatives to be at least $m(\eps_v)$ where $m(\eps_v)=\min\{m\in\NB: \rrr_\DD(m)>\eps_v\}$ due to \Cref{prop:m_lb}.

\begin{theorem}\label{cor:gen_n_test_worst_any_eps}
 Given a domain \DD, any $\eps_v$ with $0\le\eps_v<1$ with $\rrr_\DD(m)>\eps_v$, there exists a \worstcasetester{$\eps_v$}{$0$}{$\delta$} with sample complexity $\OO(\tfrac{m(\eps_v)! m(\eps_v)^2 \log^2 m(\eps_v)}{(1-(\nfrac{\eps_v}{\rrr_\DD(m(\eps_v))}))^2}\ln\tfrac{1}{\delta})$ where $m(\eps_v)=\min\{\el\in\NB: \rrr_\DD(\el)>\eps_v\}$.
\end{theorem}

\begin{proof}
 Let $\AA^\pr\subseteq\AA$ be any subset of alternatives with $|\AA^\pr|=m(\eps_v)$. We pick $\el=\tfrac{16m(\eps_v)! m(\eps_v) \log m(\eps_v)}{(1-(\nfrac{\eps_v}{\rrr_\DD(m(\eps_v))}))^2}\ln\tfrac{1}{\delta}$ preferences uniformly at random and elicit these preferences restricted to $\AA^\pr$. For ${\suc}\in\LL(\AA^\pr)$, let $X_\suc$ be the random variable denoting the number of sampled preferences which are the same as \suc. We output $1$ if $\min_{{\suc}\in\LL(\AA^\pr)} X_\suc < \tfrac{\el}{2m(\eps_v)!}(1+(\nfrac{\eps_v}{\rrr_\DD(m(\eps_v))}))$ and output $0$ otherwise. The sample complexity complexity of the algorithm is $\OO(\tfrac{m(\eps_v)! m(\eps_v)^2 \log^2 m(\eps_v)}{(1-(\nfrac{\eps_v}{\rrr_\DD(m(\eps_v))}))^2}\ln\tfrac{1}{\delta})$. The proof of correctness  of our algorithm is similar to that of \Cref{thm:spc_n_test} using the observation that, when the input profile can be made single peaked by deleting at most $\eps_v n$ preferences, there exists an ${\suc}\in\LL(\AA^\pr)$ such that $X_\suc\le \tfrac{\eps_v\el}{\rrr_\DD(m(\eps_v))m(\eps_v)!}$ (follows from the definition of $\rrr_\DD(m(\eps_v))$).
\end{proof}

We now present our result for the \worstworstcasetest{$\eps_v$}{$0$}{$\eps_v^\pr$}{$0$}{$\delta$} problem. The following structural result provides the key building block of our algorithm. Intuitively the lemma proves that, given a profile \PP, if we sample preferences from \PP uniformly at random with replacement to construct another profile \QQ (of certain size), then the ``relative'' distance of \QQ from any domain \DD is approximately same as the relative distance of \PP from \DD.

\begin{lemma}\label{lem:worst}
 Let \DD be any normal and neutral domain and $(\suc_i)_{i\in[n]}\in\LL(\AA)^n$ be a profile with preference-distance being $\eps_v n$ from \DD. Let $0<\Delta<\min\{\eps_v,1-\eps_v\}$, $\el=\tfrac{4}{\Delta^2}(\ccc_\DD(m)m! m\ln m + \ln \nfrac{1}{\delta})$, and $\suc^\pr=(\suc_i^\pr)_{i\in[\el]}$ be a profile where $\suc_i^\pr$ has been picked uniformly at random with replacement from the $n$ preferences of \suc. Then the preference-distance of $\suc^\pr$ from \DD is at least $(\eps_v-\Delta)\el$ and at most $(\eps_v+\Delta)\el$ with probability at least $1-\delta$ for every $0<\delta<1$.
\end{lemma}

\longversion{
\begin{proof}
 Let $J\subset[n]$ be a subset of $[n]$ with $|J|=\eps_v n$ such that $(\suc_i)_{i\in[n]\setminus J}\in\DD$; there exists such a $J$ since the preference-distance of \suc to \DD is $\eps_v n$. Let $X$ be the random variable denoting the number of $i\in[\el]$ such that $\suc_i^\pr=\suc_j$ for some $j\in J$. We have $\EB[X]=\eps_v\el$. Using Chernoff bound, we have the following. 
 \begin{align*}
  &\Pr[\text{preference-distance of } \suc^\pr \text{ to }\DD \text{ is less than }(\eps_v-\Delta)\el ]\\
  &\le \Pr[X>(\eps_v-\Delta)\el]\\
  &\le \exp\{-\nfrac{\Delta^2 \el}{2}\}\\
  &\le \nfrac{\delta}{2}
 \end{align*}

 We now bound the probability that the preference-distance of $\suc^\pr$ to \DD is more than  $(\eps_v+\Delta)\el$. Let $\SS\subset\LL(\AA)$ be a profile in \DD. Let $X_\SS$ be the random variable denoting the number of $i\in[\el]$ such that $\suc_i^\pr\in\SS$ and $C_\SS$ be the number of $i\in [n]$ such that $\suc_i\in\SS$. Then we have $\EB[X_\SS]=\nfrac{\el C_\SS}{n}$ and since the preference-distance of \suc to \DD is $\eps_v n$, we have $C_\SS \le (1-\eps_v)n$. Now using Chernoff bound, we have the following.
 \begin{align*}
  \Pr[X_\SS < (1-\eps_v-\Delta)\el] &\le \exp\{-\nfrac{\Delta^2 \el}{3}\}
 \end{align*}
 
 Now using union bound over all such \SS in \DD, we have the following.
 \begin{align*}
  &\Pr[\text{preference-distance of } \suc^\pr \text{ to }\DD \text{ is more than }(\eps_v+\Delta)\el]\\
  &\le \Pr[\exists \SS\subset\LL(\AA) \text{ with } \SS\in\DD, X_\SS < (1-\eps_v-\Delta)\el]\\
  &\le \ccc_\DD(m)m!\exp\{-\nfrac{\Delta^2 \el}{3}\}\\
  &\le \nfrac{\delta}{2}
 \end{align*}
 The third inequality follows from the fact that $\ccc_\DD(m)m!\le 1$ for every \DD.
\end{proof}
}

We now present our \worstworstcasetester{$\eps_v$}{$0$}{$\eps_v^\pr$}{$0$}{$\delta$}. The high level idea is to sample some number \el of preferences, compute the distance $\eps^\prr\el$ of the resulting profile from the single peaked domain, and output the distance of the original profile to be $\eps n$ if and only if $\eps^\prr$ is closer to \eps than $\eps^\pr$.

\begin{algorithm}[!htbp]
 \caption{\worstworstcasetester[single peak]{$\eps_v$}{$0$}{$\eps_v^\pr$}{$0$}{$\delta$}\label{code:pref_outlier_arb}}
  \begin{algorithmic}[1]
    \Require{Oracle access to a profile \PP}
    \Ensure{$1$ if there exists $\eps_v n$ preferences whose deletion makes the resulting profile single peaked and $0$ if deleting any $\eps_v^\pr n$ preferences from \PP does not make the resulting profile single peaked}
    \State{Sample $\el=\tfrac{64}{(\eps_v^\pr-\eps_v)^2}(2^m m\ln m + \ln \nfrac{1}{\delta})$ preferences uniformly at random from the input profile with replacement. Let $\BB\in\LL(\AA)^\el$ be the profile of sampled preferences}
    \State{Let $t$ be the minimum number of times any preference in $\LL(\{a,b,c\})$ appear in \BB}
    \If{\BB can be made single peaked by deleting at most $\nfrac{(\eps_v+\eps_v^\pr)\el}{2}$ preferences}
        \Return $1$
    \Else
        \Return $0$
    \EndIf
  \end{algorithmic}
\end{algorithm}

\begin{theorem}\label{thm:gen_n_test_worst_worst}
 For every domain \DD, there exists a \worstworstcasetester{$\eps_v$}{$0$}{$\eps_v^\pr$}{$0$}{$\delta$} with sample complexity $\OO(\tfrac{1}{(\eps_v^\pr-\eps_v)^2}(\ccc_\DD(m)m! m^2\log^2 m + \log \nfrac{1}{\delta}))$ for every $0\le \eps_v<\eps_v^\pr<1$ and $0<\delta<\nfrac{1}{2}$.
\end{theorem}

\longversion{
\begin{proof}
 \longversion{We present our algorithm in \Cref{code:pref_outlier_arb}. }We pick $\el=\tfrac{64}{(\eps_v^\pr-\eps_v)^2}(\ccc_\DD(m)m! m\ln m + \ln \nfrac{1}{\delta})$ preferences uniformly at random and query oracle to completely learn these preferences. Let the sampled profile be $\suc^\pr$. If the preference-distance of $\suc^\pr$ to \DD is at most $\nfrac{(\eps_v+\eps_v^\pr)\el}{2}$, then we output $1$; otherwise we output $0$. From \Cref{lem:worst}, we immediately have that the probability of error is at most $\delta$. The sample complexity follows from the fact that, we query each preference $\OO(m\log m)$ times to learn to completely.
\end{proof}
}

We observe that $\ccc_{\text{single peaked}}(m)=\nfrac{2^{m-1}}{m!}$~\cite[Lemma 2]{EscoffierLO08} and $\ccc_{\text{single crossing}}(m)=\nfrac{({m\choose 2}+1)}{m!}$. Hence, from \Cref{thm:gen_n_test_worst_worst}, we obtain the following result for the single peaked and single crossing domains.

\begin{corollary}\label{cor:spc_n_test_worst_worst}
 There exists a \worstworstcasetester{$\eps_v$}{$0$}{$\eps_v^\pr$}{$0$}{$\delta$} with sample complexity $\OO(\tfrac{1}{(\eps_v^\pr-\eps_v)^2}(2^m m^2\log^2 m + \log \nfrac{1}{\delta}))$ for the single peaked domain and with sample complexity $\OO(\tfrac{1}{(\eps_v^\pr-\eps_v)^2}(m^4\log^2 m + \log \nfrac{1}{\delta}))$ for the single crossing domain for every $0\le \eps_v<\eps_v^\pr<1$ and $0<\delta<\nfrac{1}{2}$.
\end{corollary}

\subsection{Only Alternatives as Outliers}

In this subsection, we now focus on the case when only alternatives are considered as outliers. We observe that when only alternatives act as outliers, the \randomtest{$0$}{$\eps_a$}{$\delta$} and \worstcasetest{$0$}{$\eps_a$}{$\delta$} are the same problem. We begin with presenting our \randomtester[single peak]{$0$}{$\eps_a$}{$\delta$} in \Cref{thm:spc_m_test} below. On a high level, our algorithm in \Cref{thm:spc_m_test} samples some number $t$ of preferences restricted to some number \el of alternatives. If for every $3$ alternatives among those \el alternatives, all the $6$ possible permutations appear in the sampled preferences, then the algorithm outputs the profile to be random, otherwise it says that the profile is close to being single peaked.

\begin{algorithm}[!htbp]
 \caption{\randomtester[single peak]{$0$}{$\eps_a$}{$\delta$}\label{code:alt_outlier}}
  \begin{algorithmic}[1]
    \Require{Oracle access to a profile \PP}
    \Ensure{$1$ if there exists $\eps_a m$ alternatives whose deletion makes the resulting profile single peaked and $0$ if \PP has been generated randomly}
    \State{Sample $\el=\min\{(1-\eps_a)m,2\log_{\nfrac{1}{\eps_a}}\nfrac{1}{\delta}\}$ alternatives uniformly at random from \AA without replacement. Let \BB be the set of sampled alternatives.}
    \State{Sample $t=18\ln \tfrac{2\log_{\nfrac{1}{\eps_a}}\nfrac{1}{\delta}}{\delta}$ uniformly random preferences restricted to \BB. Let the sampled profile be $\QQ\in\LL(\BB)^t$}
    \For{Every distinct $a, b, c\in\BB$}
        \If{at least one permutation in $\LL(\{a,b,c\})$ is not present in \QQ}
            \Return $1$
        \EndIf
    \EndFor
    \Return $0$
  \end{algorithmic}
\end{algorithm}

\begin{theorem}\label{thm:spc_m_test}
 There exists a \randomtester[single peak]{$0$}{$\eps_a$}{$\delta$} with sample complexity $\OO(\log \tfrac{\log_{\nfrac{1}{\eps_a}}\nfrac{1}{\delta}}{\delta} \log_{\nfrac{1}{\eps_a}}\tfrac{1}{\delta}\log\log_{\nfrac{1}{\eps_a}}\nfrac{1}{\delta})$. Hence, there also exists a \worstcasetester[single peak]{$0$}{$\eps_a$}{$\delta$} with the same sample complexity for every $0<\eps_a<1$ and $0<\delta<\nfrac{1}{2}$ such that $\ccc_{\text{single peak}}((1-\eps_a)m)<1$.
\end{theorem}

\begin{proof}
 \longversion{We present our algorithm in \Cref{code:alt_outlier}. }We sample $\el=\min\{(1-\eps_a)m,2\log_{\nfrac{1}{\eps_a}}\nfrac{1}{\delta}\}$ alternatives uniformly at random without replacement. Let \BB be the set of sampled alternatives. We now sample $t=18\ln \tfrac{2\log_{\nfrac{1}{\eps_a}}\nfrac{1}{\delta}}{\delta}$ preferences uniformly at random with replacement restricted to \BB. Let $\QQ$ be the set of sampled preferences. We output $1$ if there exist $3$ alternatives $a,b,c\in\BB$ such that at least one permutation in $\LL(\{a,b,c\})$ is not present in \QQ and output $0$ otherwise. We observe that the sample complexity of our algorithm is $\OO(t\el\log\el)=\OO(\log \tfrac{\log_{\nfrac{1}{\eps_a}} \nfrac{1}{\delta}}{\delta} \log_{\nfrac{1}{\eps_a}}\tfrac{1}{\delta}\log\log_{\nfrac{1}{\eps_a}}\nfrac{1}{\delta})$. We now turn to the correctness of our algorithm. For that we show that irrespective of the input profile, the probability of making an error is at most $\delta$. 
 \begin{itemize}[leftmargin=0cm,itemindent=0.3cm,labelwidth=\itemindent,labelsep=0cm,align=left,noitemsep,topsep=2pt]
  \item {\bf Case I - the input profile is single peaked after deleting at most $\eps_a m$ alternatives:} Let \AA be the set of alternatives and $\WW\subset\AA$ with $|\WW|\le \eps_a m$ such that the input profile restricted to $(\AA\setminus\WW)$ is single peaked. Then we have the following for the chosen value of \el:
  \begin{align*}
   \Pr[\text{error}] &\le \Pr[|\BB\cap\WW|\ge\el-2]\\
   &= \eps_a^\el + {\el\choose 1}(1-\eps_a)\eps_a^{\el-1} + {\el\choose 2}(1-\eps_a)^2\eps_a^{\el-2}\\
   &\le \eps_a^\el + \el\eps_a^{\el-1} + \el^2\eps_a^{\el-2}\le \delta
  \end{align*}
  
  \item {\bf Case II - the input profile has been generated uniformly at random:} For any $3$ alternatives $a,b,c\in\AA$, we define a random variable $X_{\{a,b,c\}}$ to be $1$ if all $6$ possible permutations in $\LL(\{a,b,c\})$ are present in $\QQ(\{a,b,c\})$ and $0$ otherwise. Using folklore tail bound for the {\em coupon collector problem} (for example, see \cite[Chap 3.6]{motwani2010randomized}), we obtain the following for the chosen value of $t$.
  \begin{align*}
   \Pr[X_{\{a,b,c\}}=0] \le 6^{-\nfrac{t}{3\ln 6}} \le e^{-\nfrac{t}{6}}
  \end{align*}

  Now using union bound, we obtain the following for the chosen values of \el and $t$.
  \begin{align*}
   \Pr[\text{error}] &\le \Pr[\exists\{a,b,c\}\subset\BB, X_{\{a,b,c\}}=0]\le {\el\choose 3} e^{-\tfrac{t}{6}}\le \delta\qedhere
  \end{align*}
 \end{itemize}
\end{proof}

From the proof of \Cref{thm:spc_m_test}, \Cref{cor:gen_m_test} follows.

\begin{corollary}\label{cor:gen_m_test}
 For every domain \DD, there exists a \worstcasetester{$0$}{$\eps_a$}{$\delta$} with sample complexity $\OO(\log \tfrac{\log_{\nfrac{1}{\eps_a}}\nfrac{1}{\delta}}{\delta} \log_{\nfrac{1}{\eps_a}}\tfrac{1}{\delta}\log\log_{\nfrac{1}{\eps_a}}\nfrac{1}{\delta})$ for every $0<\eps_a<1$ and $0<\delta<\nfrac{1}{2}$ such that $\ccc_{\DD}((1-\eps_a)m)<1$.
\end{corollary}

We show below that the condition $\ccc_{\DD}((1-\eps_a)m)<1$ in \Cref{thm:spc_m_test,cor:gen_m_test} is necessary. We prove \Cref{prop:gen_m_imposs} by carefully constructing a set of outliers such the the sample distribution in both the possibilities are statistically indistinguishable.

\begin{proposition}\label{prop:gen_m_imposs}
 For every domain \DD, there does not exist any \worstcasetester{$0$}{$\eps_a$}{$\delta$} if $\ccc_{\DD}((1-\eps_a)m)=1$.
\end{proposition}

\longversion{
\begin{proof}
 Follows immediately from the observation that any profile (in particular a profile where every preference is distributed uniformly in $\LL(\AA)$) can be made to belong to the domain \DD by deleting any $\eps_a m$ alternatives.
\end{proof}
}

We now turn to the \worstworstcasetest[single peak]{$0$}{$0$}{$0$}{$\eps_a^\pr$}{$\delta$} problem. The following results show that the sample complexity of this problem is $\Omega(n\log \nfrac{1}{\delta})$ even for the single peaked and single crossing domains.

\begin{theorem}\label{thm:sp_m_test_worst_worst}
 Any \worstworstcasetester[single peak]{$0$}{$0$}{$0$}{$\eps_a^\pr$}{$\delta$} has sample complexity $\Omega(n\log \nfrac{1}{\delta})$ for every $0<\eps_a^\pr\le 1$ and $0<\delta<\nfrac{1}{2}$ such that $\ccc_{\text{single peak}}((1-\eps_a^\pr)m)<1$.
\end{theorem}

\longversion{
\begin{proof}
 Let us consider the profile $\suc=(\suc_i)_{i\in[n]}$ on the set $\AA=\{a_j, b_j, c_j:j\in[m]\}$ of alternatives defined as follows for any even integer $n$.
 \begin{align*}
  i\in [\nfrac{n}{2}], \suc_i &: a_1\suc b_1\suc c_1\suc \longversion{\cdots a_j\suc b_j\suc c_j\suc} \cdots a_m\suc b_m\suc c_m\\
  i\in \left[\frac{n}{2}+1,n-1\right], \suc_i &: c_m\suc b_m\suc a_m \suc \longversion{\cdots c_j\suc b_j\suc a_j\suc}\cdots c_1\suc b_1\suc a_1\\
  \suc_n &: a_1\suc c_1\suc b_1\suc \longversion{\cdots a_j\suc c_j\suc b_j\suc} \cdots a_m\suc c_m\suc b_m\\
 \end{align*}
 
 We observe that at least one alternative in $\{a_i, b_i, c_i\}$ for every $i\in[m]$ needs to be deleted to make the resulting profile single peaked. Hence the alternative-distance of \suc from single peaked domain is $m$. Now let us consider another profile $\suc^\pr=(\suc_i^\pr)_{i\in[n]}$ on \AA which is defined as follows.
 \[\suc_i^\pr=\suc_i \forall i\in[n-1] \text{ and } \suc_n^\pr=\suc_1\]
 
 The profile $\suc^\pr$ is single peaked with respect to the societal order $a_1\suc b_1\suc c_1\suc \longversion{\cdots a_j\suc b_j\suc c_j\suc} \cdots a_m\suc b_m\suc c_m$ of the set of alternatives. Now any \worstworstcasetester{$0$}{$0$}{$0$}{$\nfrac{1}{3}$}{$\delta$} would be able to distinguish \suc from $\suc^\pr$ with probability at least $1-\delta$ which needs $\Omega(n\log \nfrac{1}{\delta})$ queries. This proves the statement for every $\eps_a^\pr\ge \nfrac{1}{3}$. To prove the result for any $0<\eps_a^\pr<\nfrac{1}{3}$, we add $\nfrac{m}{\eps_a^\pr}$ dummy alternatives at the bottom of every preference in the same order.
\end{proof}
}

\longversion{ We have similar lower bound for the single crossing domain too in \Cref{thm:sc_m_test_worst_worst}.}

\begin{theorem}\label{thm:sc_m_test_worst_worst}
 Any \worstworstcasetester{$0$}{$0$}{$0$}{$\eps_a^\pr$}{$\delta$} has sample complexity $\Omega(n\log \nfrac{1}{\delta})$ for single crossing domain for every $0<\eps_a^\pr\le 1$ and $0<\delta<\nfrac{1}{2}$ such that $\ccc_{\text{single crossing}}((1-\eps_a^\pr)m)<1$.
\end{theorem}

\longversion{
\begin{proof}
 Let us consider the profile $\suc=(\suc_i)_{i\in[n]}$ on the set $\AA=\{a_j, b_j, c_j:j\in[m]\}$ of alternatives defined as follows for any integer $n$ divisible by $4$.
 \begin{align*}
  i\in[1,\nfrac{n}{4}], \suc_i &: a_1\suc b_1\suc c_1\suc \longversion{\cdots a_j\suc b_j\suc c_j\suc} \cdots a_m\suc b_m\suc c_m\\
  i\in [\nfrac{n}{4}+1,\nfrac{n}{2}], \suc_i &: a_1\suc c_1\suc b_1\suc \longversion{\cdots a_j\suc c_j\suc b_j\suc} \cdots a_m\suc c_m\suc b_m\\  
  i\in [\nfrac{n}{2}+1,\nfrac{3n}{4}], \suc_i &: c_1\suc a_1\suc b_1\suc \longversion{\cdots c_j\suc a_j\suc b_j\suc} \cdots c_m\suc a_m\suc b_m\\  
  i\in [\nfrac{3n}{4}+1,n-1], \suc_i &: c_1\suc b_1\suc a_1\suc \longversion{\cdots c_j\suc b_j\suc a_j\suc} \cdots c_m\suc b_m\suc a_m\\  
  \suc_n &: b_1\suc a_1\suc c_1\suc \longversion{\cdots b_j\suc a_j\suc c_j\suc} \cdots b_m\suc a_m\suc c_m\\
 \end{align*}
 
 We observe that at least one alternative in $\{a_i, b_i, c_i\}$ for every $i\in[m]$ needs to be deleted to make the resulting profile single crossing. Hence the alternative-distance of \suc from single crossing domain is $m$. Now let us consider another profile $\suc^\pr=(\suc_i^\pr)_{i\in[n]}$ on \AA which is defined as follows.
 \[\suc_i^\pr=\suc_i \forall i\in[n-1] \text{ and } \suc_n^\pr=\suc_1\]
 
 The profile $\suc^\pr$ is single crossing with respect to the preference ordering $(\suc_i)_{i\in[n-1]}$. Now any \worstworstcasetester{$0$}{$0$}{$0$}{$\eps_a^\pr$}{$\delta$} would be able to distinguish \suc from $\suc^\pr$ with probability at least $1-\delta$ which needs $\Omega(n\log \nfrac{1}{\delta})$ queries. This proves the statement for every $\eps_a^\pr\ge \nfrac{1}{3}$. To prove the result for any $0<\eps_a^\pr<\nfrac{1}{3}$, we add $\nfrac{m}{\eps_a^\pr}$ dummy alternatives at the bottom of every preference in the same order.
\end{proof}
}

\longversion{
\subsection{Both Preferences and Alternatives as Outliers}

In this section, we consider the case when both preferences and alternatives can simultaneously be considered as outliers.
\begin{theorem}\label{thm:spc_mn_test}
 There exists a \randomtester[single peak]{$\eps_v$}{$\eps_a$}{$\delta$} with sample complexity $\OO(\tfrac{1}{(1-\eps_v)^2}\log \tfrac{\log_{\nfrac{1}{\eps_a}}\nfrac{1}{\delta}}{\delta} \log_{\nfrac{1}{\eps_a}}\tfrac{1}{\delta}\log\log_{\nfrac{1}{\eps_a}}\nfrac{1}{\delta})$ for every $0\le\eps_v<1$, $0<\eps_a<1$, and $0<\delta<\nfrac{1}{2}$ such that $\ccc_{\text{single peak}}((1-\eps_a)m)<1$.
\end{theorem}

\longversion{
\begin{proof}
 We sample $\el=\min\{(1-\eps_a)m,2\log_{\nfrac{1}{\eps_a}}\nfrac{1}{\delta}\}$ alternatives uniformly at random with replacement. Let \BB be the sampled set of alternatives. We now sample $t=\tfrac{144}{(1-\eps_v)^2}\ln \nfrac{\el}{\delta}$ preferences uniformly at random with replacement restricted to \BB. Let $\QQ$ be the set of sampled preferences. For any $3$ distinct alternatives $a,b,c\in\BB$, let $X_{a> b> c} = |\{{\suc}\in\QQ: a\suc b\suc c\}|$; that is $X_{a> b> c}$ is the random variable denoting the number of preferences in \QQ where $a$ is preferred over $b$ and $b$ is preferred over $c$. Let us define another random variable $Y$ as $Y = \min_{a,b,c\in\BB} X_{a>b>c}$. We output $1$ if $Y < \tfrac{t}{12}(1+\eps)$ and output $0$ otherwise. We observe that the sample complexity of our algorithm is $\OO(t\el\log\el)=\OO(\tfrac{1}{(1-\eps_v)^2}\log \tfrac{\log_{\nfrac{1}{\eps_a}}\nfrac{1}{\delta}}{\delta} \log_{\nfrac{1}{\eps_a}}\tfrac{1}{\delta}\log\log_{\nfrac{1}{\eps_a}}\nfrac{1}{\delta})$. We now turn to the correctness of our algorithm. For that we show that irrespective of the input profile, the probability of making an error is at most $\delta$.
 \begin{itemize}[leftmargin=0cm,itemindent=0.3cm,labelwidth=\itemindent,labelsep=0cm,align=left,noitemsep,topsep=2pt]
  \item {\bf Case I - the input profile is single peaked after deleting at most $\eps_a m$ alternatives and $\eps_v n$ preferences (which are generated randomly from the set of all preferences over the rest of the alternatives):} Let \AA be the set of alternatives and $\WW\subset\AA$ a subset of \AA with $|\WW|\le \eps_a m$ such that the input profile restricted to $(\AA\setminus\WW)$ is single peaked after deleting at most $\eps_v n$ preferences. Using the calculation in \Cref{thm:spc_m_test}, we have $\Pr[|\BB\cap\WW|\ge\el-2]\le\nfrac{\delta}{2}$. Since at least $(1-\eps_v)n$ preferences form a single peaked profile and the rest (at most $\eps_v n$) of the preferences have been chosen uniformly at random with replacement from $\LL(\AA\setminus\WW)$, there exists distinct alternatives $a,b,c\in\AA\setminus\WW$ and ${\suc}\in\LL(\{a,b,c\})$ such that $\EB[X_\suc]=\nfrac{\eps_v t}{6}$. Using Chernoff bound, we have $\Pr[X_i \ge \tfrac{t}{12}(1+\eps_v)] \le \exp\{-\tfrac{(1-\eps_v)^2}{\eps_v}\tfrac{t}{144}\} \le \nfrac{\delta}{2}$. Using union bound, we now bound the probability of error in this case as follows.
  \[\Pr[\text{error}] \le \Pr[|\BB\cap\WW|\ge\el-2] + \Pr[X_i \ge \tfrac{t}{12}(1+\eps_v)]\le\delta\]
  
  \item {\bf Case II - the input profile has been generated uniformly at random:} For every $a,b,c\in\BB$, we have $\EB[X_{a>b>c}]=\nfrac{t}{6}$. Using Chernoff bound followed by union bound, we have the following:
  \begin{align*}
   \Pr[\text{error}] &= \Pr[\exists a,b,c\in\BB, X_{a>b>c}\le \tfrac{t}{12}(1+\eps_v)]\\
   &\le 6{\el\choose 3}\exp\{-\nfrac{(1-\eps_v)^2 t}{48}\}\\
   &\le \delta
  \end{align*}

 \end{itemize}
\end{proof}
}

From the proof of \Cref{thm:spc_mn_test}, the following more general result is immediate using similar argument as in \Cref{cor:gen_n_test}\longversion{ and \Cref{cor:gen_m_test}}.

\begin{corollary}\label{cor:gen_mn_test}
 For every domain \DD, there exists a \randomtester{$\eps_v$}{$\eps_a$}{$\delta$} with sample complexity $\OO(\tfrac{1}{(1-\eps_v)^2}\log \tfrac{\log_{\nfrac{1}{\eps_a}}\nfrac{1}{\delta}}{\delta} \log_{\nfrac{1}{\eps_a}}\tfrac{1}{\delta}\log\log_{\nfrac{1}{\eps_a}}\nfrac{1}{\delta})$ for every $0\le\eps_v<1$, $0<\eps_a<1$, and $0<\delta<\nfrac{1}{2}$ such that $\ccc_{\DD}((1-\eps_a)m)<1$ such that $\ccc_{\DD}((1-\eps_a)m)<1$.
\end{corollary}

From the proof of \Cref{thm:spc_mn_test} and \Cref{thm:spc_n_test_worst}, we derive the following corollary.

\begin{corollary}\label{cor:spc_mn_test_worst}
 There exists a \worstcasetester[single peak]{$\eps_v$}{$\eps_a$}{$\delta$} with sample complexity $\OO(\tfrac{1}{(1-3\eps_v)^2}\log \tfrac{\log_{\nfrac{1}{\eps_a}}\nfrac{1}{\delta}}{\delta} \log_{\nfrac{1}{\eps_a}}\tfrac{1}{\delta}\log\log_{\nfrac{1}{\eps_a}}\nfrac{1}{\delta})$ for every $0\le\eps_v<\nfrac{1}{3}$ such that $\ccc_{\text{single peak}}((1-\eps_a)m)<1$.
\end{corollary}

\longversion{
\begin{proof}
 Our algorithm for this case is similar to Theorem 8 except the following. We use $t=\tfrac{144}{(1-3\eps_v)^2}\ln \nfrac{\el}{\delta}$; the value of \el remains same. We output $1$ if $Y < \tfrac{\el}{12}(1+3\eps_v)$ and output $0$ otherwise (with notations as defined in Theorem 1). The proof of correctness and sample complexity of our algorithm is similar to the correctness of Theorem 1 using the observation that, when the input profile can be made single peaked by deleting at most $\eps_v n$ preferences and $\eps_a m$ alternatives, there exists an $i\in[6]$ such that $\EB[X_i]=\nfrac{\eps_v t}{2}$.
\end{proof}
}

The following generalization of \Cref{cor:spc_mn_test_worst} is also immediate.

\begin{corollary}\label{cor:gen_mn_test_worst}
 Let $m_0=\min\{m\in\NB: \rrr_\DD(m)<1\}$. Then there exists a \worstcasetester{$\eps_v$}{$\eps_a$}{$\delta$} with sample complexity $\OO(\tfrac{1}{(1-(\nfrac{\eps_v}{\rrr_\DD(m_0)}))^2}\log \tfrac{\log_{\nfrac{1}{\eps_a}}\nfrac{1}{\delta}}{\delta} \log_{\nfrac{1}{\eps_a}}\tfrac{1}{\delta}\log\log_{\nfrac{1}{\eps_a}}\nfrac{1}{\delta})$ for every $\eps_v<\rrr_\DD(m_0)$ such that $\ccc_{\DD}((1-\eps_a)m)<1$.
\end{corollary}
}

%% file: empirical.tex
\section{Empirical evaluation}

The algorithms presented in \Cref{sec:results} provide upper bounds on the sample complexities of the problems of outlier detection. These algorithms distinguish between two possibilities of profile generation with a probability of correctness of at least $(1-\delta)$. 
It is interesting to find out the optimal multiplying factors of the sampling complexities inside $\OO(\cdot)$\longversion{ that drive the classification probabilities beyond the threshold of $(1-\delta)$} in these algorithms. This is why an empirical evaluation is called for. 

In this section, we empirically find the factors for the results of \Cref{thm:spc_n_test,thm:spc_n_test_worst,thm:spc_m_test}, which provide constant time algorithms for the testing problem. The other two cases as shown in \Cref{tbl:summary} either consider an exponential time (\Cref{cor:spc_n_test_worst_worst}) algorithm or provide a lower bound\longversion{ that is linear in $n$} (\Cref{thm:sp_m_test_worst_worst}), which are unsuitable for an empirical study.

\subsection{Approach for \Cref{thm:spc_n_test,thm:spc_n_test_worst}:} We generate $n=10,000$ preferences with $m$ alternatives uniformly at random to form a preference profile. The sampling algorithm of \Cref{thm:spc_n_test} (given by \Cref{code:pref_outlier}) picks an $\ell$ for a given $\epsilon_v$. In this experiment, we choose a sampling size $l$ that is smaller than $\ell$, and apply the same algorithm using $l$ preferences sampled with replacement from the population of $n$. We generate the preference profile $100$ times and for every profile, sample $l$ preferences $100$ times. We consider the fraction of correct classifications given by this modified sampling algorithm and plot it with increasing $l$. We fix $\delta=0.001$ for these evaluations. We show the plot of the fraction of correct classification (denoted by $\rho$) for \Cref{thm:spc_n_test} with $m=3$ in \Cref{thm1_profile2_normalized}. 
%
%
%
\begin{figure}[h!]
\centering
 \includegraphics[width=1.1\linewidth]{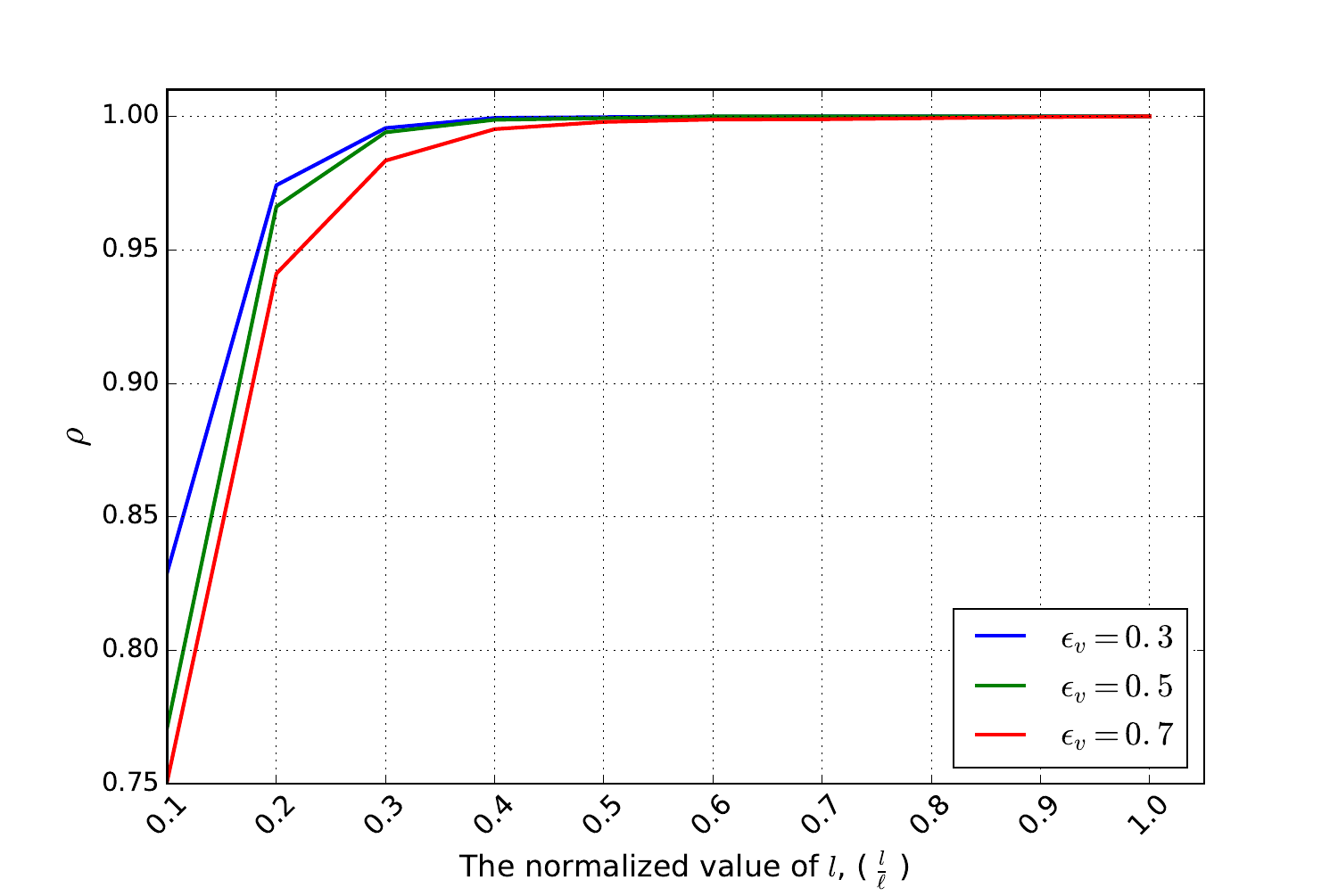}
 \caption{Fraction of correct classification ($\rho$) of the adaptation of \Cref{code:pref_outlier} when $l (\le\ell)$ preferences have been sampled 
 uniformly at random from a random preference profile of size $n=10,000$, $\delta = 0.001$ 
 (x-axis shows the normalized value, $\nfrac{l}{\ell}$).}
 \label{thm1_profile2_normalized}
\end{figure}
%
The plot shows the growth of the empirical probability of correctness (and therefore does not need any errorbar). 
The x-axis shows the normalized sample size (that is $\nfrac{l}{\ell}$). Notice that the growth of the curves almost overlaps for different $\epsilon_v$s, 
and reaches $(1-\delta)$ nearly at $0.5$. This empirically shows that when other parameters are held fixed at the chosen values, the hidden constant in the upper bound of the sample complexity in the context of random outliers can be reduced by almost 50\%, and is independent of $\epsilon_v$. 


We perform a similar exercise with different sampling sizes for the algorithm in the proof of \Cref{thm:spc_n_test_worst} (given by \Cref{code:pref_outlier_arb}) with $m=5$ in \Cref{thm2_profile2_normalized}. Here too, the proportionality factor is independent of the $\epsilon_v$s, and the hidden constant factor in this case can be reduced by 60\%.
\begin{figure}[h!]
\centering
 \includegraphics[width=1.1\linewidth]{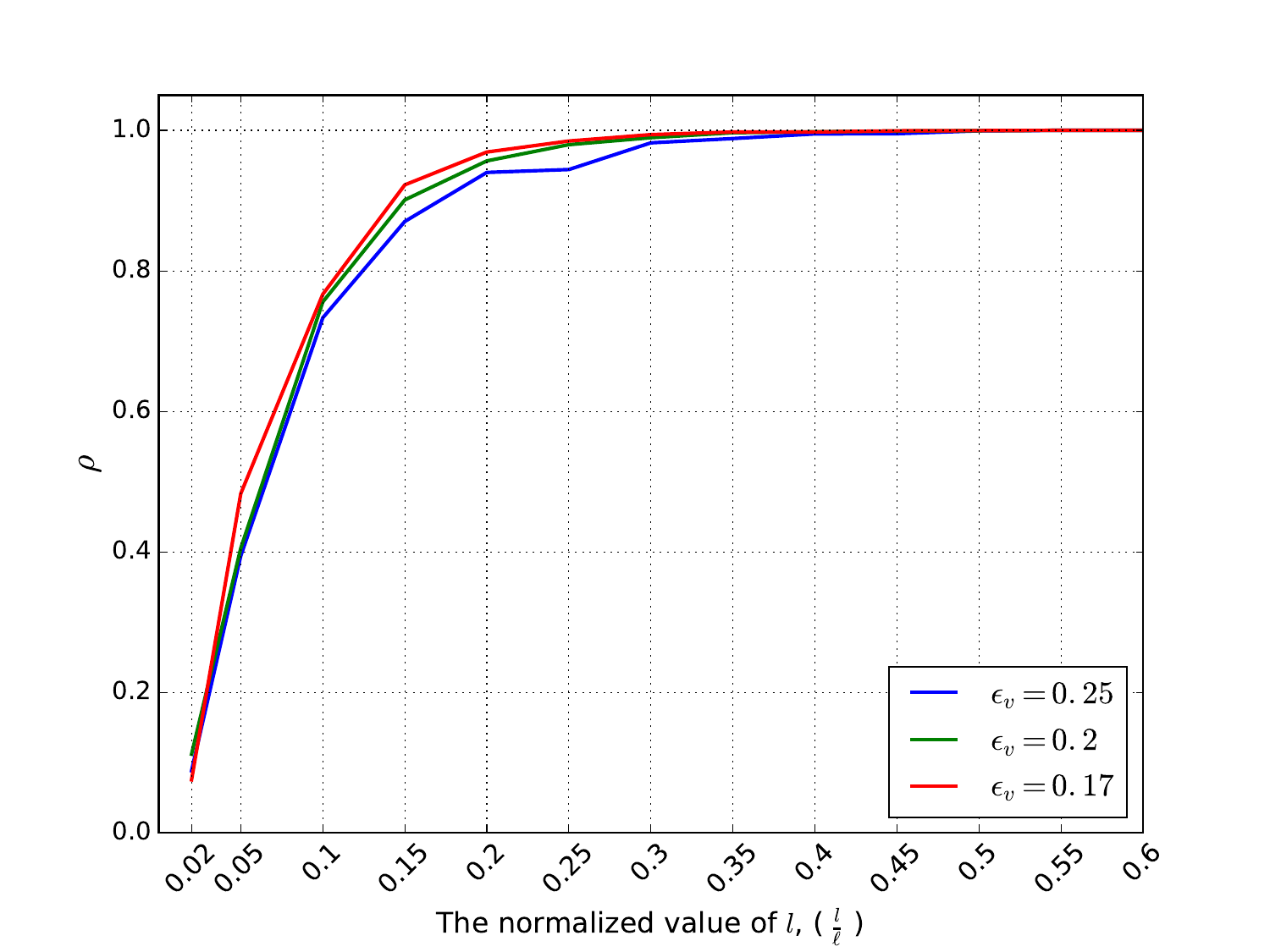}
 \caption{Fraction of correct classification ($\rho$) of the adaptation of \Cref{code:pref_outlier_arb} when $l (\leqslant \ell)$ preferences have been sampled 
 uniformly at random from a random preference profile of size $n=10,000$, $\delta = 0.001$ 
 (x-axis shows the normalized value, $\nfrac{l}{\ell}$).}
 \label{thm2_profile2_normalized}
\end{figure}

\smallskip
\paragraph{\em Why the error with a random/arbitrary outliers profile being classified as a random profile is not considered?}

We argue that such an error is not very likely in the algorithms of these theorems, which is also manifested in our simulations. Therefore we omit them presenting here. For \Cref{thm:spc_n_test}, since the focus is only on the three alternatives $a,b$, and $c$, the number of random outliers will be close to $\nfrac{\epsilon_v n}{6}$ for large enough $n$. If $l$ preferences are drawn uniformly at random with replacement from this profile, it is very likely that $\min_{i} X_i$ will be at most close to $\nfrac{\epsilon_v l}{6}$ for reasonably sized $l$. The algorithm classifies the profile as random outlier profile if $\min_{i} X_i \leqslant \nfrac{l (1+\epsilon_v)}{12}$ and since $\nfrac{\epsilon_v l}{6}\leqslant \nfrac{l (1+\epsilon_v)}{12}$, it is unlikely that a random outlier profile will be classified as random profile under this algorithm. Similar observation is true for \Cref{thm:spc_n_test_worst}.


\subsection{Approach for \Cref{thm:spc_m_test}:} Here we consider the alternatives as outliers. The algorithm in the proof of this theorem (given by \Cref{code:alt_outlier}) samples $\ell$ alternatives uniformly at random and samples $t$ preferences restricted to the sampled alternatives uniformly at random. In this case, we pick the values of $\delta$ and $n$ as before. We fix $m=9$, and pick $\ell = \min\{(1-\eps_a)m,2\log_{\nfrac{1}{\eps_a}}\nfrac{1}{\delta}\}$ as given in the proof of \Cref{thm:spc_m_test}, and vary the value of $\tau (\leqslant t)$, which is the sampling size of the preferences restricted to the chosen $\ell$ alternatives.
The alternatives of size $\ell$ are sampled 100 times.
\Cref{thm5_profile2_normalized} shows the plot of the fraction of correct classification ($\rho$) under this setting. It empirically shows that when other parameters are held fixed at the chosen values, the hidden constant of the upper bound of the probability in the case of random alternative outliers can be reduced by almost 75\%, and is independent of $\epsilon_a$.
\begin{figure}[h!]
\centering
 \includegraphics[width=1.1\linewidth]{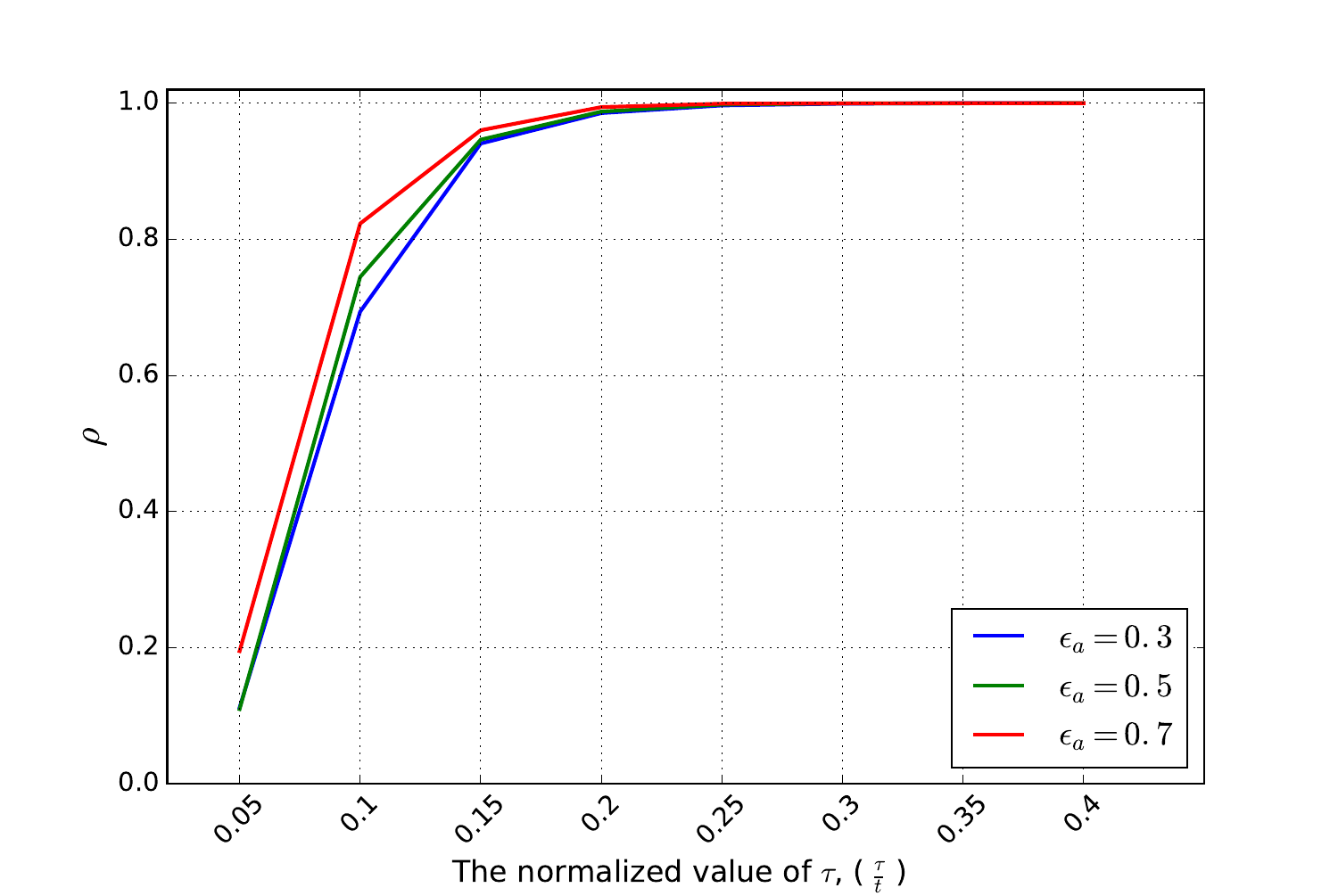}
 \caption{Fraction of correct classification ($\rho$) of the adaptation of \Cref{code:alt_outlier} when $\tau (\leqslant t)$ preferences have been sampled
 uniformly at random from a random preference profile of size $n=10,000$ with $\ell = \min\{(1-\eps_a)m,2\log_{\nfrac{1}{\eps_a}}\nfrac{1}{\delta}\}$ as given in the proof, $m=10, \delta = 0.001$ 
 (x-axis shows the normalized value, $\nfrac{\tau}{t}$).}
 \label{thm5_profile2_normalized}
\end{figure}

In a way similar to the previous paragraph we can argue that this algorithm also has a bias towards classifying a profile as random alternative outlier, which also is empirically manifested. Hence, we omit presenting them here.

%% file: conclusion.tex
\section{Discussion}

In this paper, we have developed sampling based algorithms for testing if a profile is close to some specific domain. These testing problem can be quite accurately solved by observing a small number of samples for most of the cases, and the numbers are often independent to the number of preferences or alternatives. In other cases, we have proved impossibility results. Our extensive empirical study further improve the constants of the asymptotic theoretical upper bounds on the sample complexity by 50\% to 75\% depending on the problem. As a future work, there exist more sophisticated notion of distances, namely swap distance, footrule distance, maximum displacement distance, etc. where it will be interesting to extend our results to those fine grained measures of distance.